\documentclass[sigconf]{acmart}

\usepackage{multirow}
\usepackage{enumitem}
\usepackage{algorithm}
\usepackage{algorithmic}
\usepackage{amsthm}
\usepackage{framed}
\usepackage{siunitx}
\newtheorem{proposition}{Proposition}
\newtheorem{lemma}{Lemma}

\definecolor{pboxframe}{gray}{0.25}
\definecolor{pboxback}{gray}{0.95}
\newcounter{protocolbox}
\renewcommand{\theprotocolbox}{\arabic{protocolbox}}
\makeatletter
\newenvironment{protocolbox}[1]{%
  \par\addvspace{\medskipamount}%
  \let\@vspace\@vspace@orig \let\@vspacer\@vspacer@orig
  \refstepcounter{protocolbox}%
  \noindent\colorbox{pboxframe}{%
    \parbox{\dimexpr\columnwidth-2\fboxsep\relax}{%
      \color{white}\bfseries\small Box~\theprotocolbox: #1}}%
  \par\nointerlineskip
  \MakeFramed{\advance\hsize-\width\FrameRestore}%
  \small\ignorespaces
}{%
  \endMakeFramed\par\addvspace{\medskipamount}%
}
\makeatother

\newcommand{\Csub}{C_{I}^{\mathrm{sub}}}
\newcommand{\Wtilde}{\widetilde{W}}
\newcommand{\Eitem}{E_{\mathrm{item}}}

\newcommand{\Xitem}{X_{\mathrm{item}}}
\newcommand{\Xuser}{X_{\mathrm{user}}}
\newcommand{\Eitemhat}{\widehat{E}_{\mathrm{item}}}
\newcommand{\Euserhat}{\widehat{E}_{\mathrm{user}}}

\AtBeginDocument{%
  }

\copyrightyear{2026}
\acmYear{2026}
\setcopyright{cc}
\setcctype{by}
\acmConference[CIKM '26]{Proceedings of the 35th ACM International Conference on Information and Knowledge Management}{November 07--11, 2026}{Rome, Italy}
\acmBooktitle{Proceedings of the 35th ACM International Conference on Information and Knowledge Management (CIKM '26), November 07--11, 2026, Rome, Italy}
\acmDOI{10.1145/3799682.3840871}
\acmISBN{979-8-4007-2539-5/2026/11}
\begin{document}

\title{The Edge Spectrum of Choice-Derived Item Graphs: Strong and\\ Weak Edges Encode Different Relations in Collaborative Filtering}

\author{Keigo Sakurai}
\orcid{0009-0008-3747-0491}
\affiliation{%
  \institution{Hokkaido University}
  \city{Sapporo}
  \country{Japan}}
\email{sakurai@lmd.ist.hokudai.ac.jp}

\author{Takahiro Ogawa}
\orcid{0000-0003-2644-2715}
\affiliation{%
 \institution{Hokkaido University}
 \city{Sapporo}
  \country{Japan}}
 \email{ogawa@lmd.ist.hokudai.ac.jp}

\author{Miki Haseyama}
\orcid{0000-0003-1496-1761}
\affiliation{%
  \institution{Hokkaido University}
  \city{Sapporo}
  \country{Japan}}
  \email{mhaseyama@lmd.ist.hokudai.ac.jp}

\begin{abstract}
Graph collaborative filtering relies on item--item graphs whose edges are used for positive smoothing, under the implicit assumption that stronger edges encode \emph{more of the same} relation as weaker ones.
We show that this assumption fails for a practically important class of graphs: those whose edge weights come from a choice model.
On such graphs, strong and weak edges encode \emph{qualitatively different} relations, which we call an \emph{edge spectrum}.
Specifically, strong edges concentrate on the in-slate competitors of clicked items, exactly the pairs that the within-slate ranking gradient pushes apart, while weak edges do not. We formalize this as a sign mismatch between the smoothing operator and the ranking gradient, and prove that co-click graphs cannot exhibit the same misalignment by construction.
This diagnosis explains three empirical observations on MIND and EB-NeRD: (i) drop-in choice-derived operators do not beat co-click, despite indexing structurally distinct neighborhoods; (ii) uniform scalar fixes (sign flip, in-slate margin loss) fail predictably, because the misalignment lives in the graph, not in the loss; (iii) only edge-magnitude-aware operators, with the regime boundary located by the diagnosis rather than by tuning, recover the predicted ordering. The neighbor cutoff $k$ is therefore a \emph{semantic switch}, not a sparsification hyperparameter.
Our claim concerns \emph{which interventions fail or succeed and why}, not absolute headline gains, which the diagnosis itself predicts to be small under the attenuated propagation channel we observe. We turn the diagnosis into a reusable protocol practitioners can run before deploying any choice-derived item-side operator.
Code: \url{https://github.com/kyomusso/Edge-Spectrum-in-CF}.
\end{abstract}



\maketitle

\section{Introduction}
\label{sec:intro}
Item--item graphs are everywhere in graph collaborative filtering. LightGCN~\cite{model/he2020lightgcn}'s even-hop terms implicitly use one; ItemKNN~\cite{model/sarwar2001item}, SLIM~\cite{model/ning2011slim}, EASE~\cite{model/steck2019ease}, and GFCF~\cite{model/shen2021gfcf} use one explicitly; methods that augment recommendation with knowledge graphs~\cite{model/wang2019kgat}, item content~\cite{model/he2016vbpr}, social ties~\cite{model/fan2019graph}, or self-supervised augmentation~\cite{model/wu2021sgl} all use one. Across this design space, a single assumption is shared: a stronger edge means \emph{more of the same} relation as a weaker edge, just with higher confidence. Under this view, the neighbor cutoff $k$ is a sparsification knob: a choice of how much of that one fixed relation to retain.

This paper shows that the assumption fails for an important class of graphs: those whose edge weights come from a choice model~\cite{model/lcm4rec,model/ccf,method/aouad2023assortment}, the natural way to encode in-slate competition between items. On such graphs, the edge weight $W_{ij}$ does more than measure confidence; it carries semantic content. We call the high-weight head of the distribution \emph{strong edges}: pairs of items that consistently competed for the same click within the same slate. We call the low-weight tail \emph{weak edges}: pairs that merely co-appeared occasionally without strong competitive structure. The central observation of this paper is that strong edges and weak edges encode \emph{qualitatively different relations}, so two values of the same neighbor cutoff $k$ can index relations of opposite sign. We call this an \emph{edge spectrum}.

\begin{figure}[t]
  \centering
  \includegraphics[width=\columnwidth]{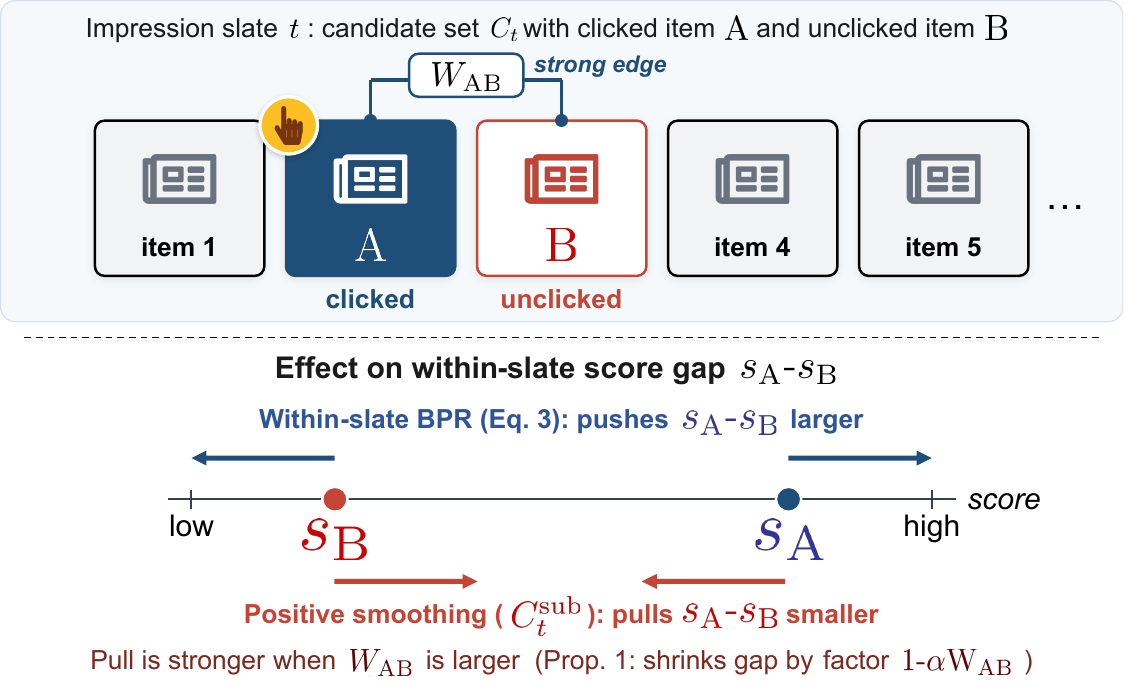}
  \caption{The edge spectrum: choice-derived substitution edges connect a clicked item to its strongest in-slate competitors. Positive smoothing over such edges raises the competitor's score; the within-slate ranking gradient (Eq.~\ref{eq:bpr}) pushes in the opposite direction (Proposition~\ref{prop:sign}). Strong edges and weak edges encode different relations.}
  \label{fig:concept}
  \Description{A schematic showing a clicked item connected by a strong choice-derived substitution edge to its strongest in-slate competitor, with the positive smoothing operator and the within-slate ranking gradient acting on that pair in opposite directions, alongside a weak edge that does not show this behavior.}
  \vspace{-0.4cm}
\end{figure}

Figure~\ref{fig:concept} shows the mechanism. Consider a clicked item $A$ and an unclicked item $B$ in the same impression slate. The within-slate ranking gradient pushes their scores \emph{apart}; this is the entire point of impression-aware training. Now suppose $A$ and $B$ form a \emph{strong} edge in the choice-derived graph: by construction, this happens precisely \emph{because} they are competitors for the same click. Positive smoothing over that strong edge then pulls their scores \emph{together}, in direct opposition to the ranking gradient. \emph{Weak} edges do not exhibit this property: they neither concentrate on in-slate competitor pairs nor produce a comparable pull. We formalize the strong-edge case as a sign mismatch (Proposition~\ref{prop:sign}). Co-click graphs cannot exhibit the same misalignment by construction: pairs with one clicked and one unclicked item contribute \emph{zero} weight to them (Lemma~\ref{lemma:coclick}).

We make three claims, each tested on MIND and EB-NeRD:
\begin{itemize}[leftmargin=*,topsep=2pt,itemsep=1pt]
  \item \textbf{Drop-in fails (\S\ref{sec:main-comparison}).} Choice-derived operators do not improve over co-click, despite indexing largely disjoint neighborhoods.
  \item \textbf{Scalar fixes fail predictably (\S\ref{subsec:scalar-remedies}).} Sign-flipping the operator or adding a competition-aware margin loss does not work: the misalignment lives in the graph, not in the loss.
  \item \textbf{Edge-magnitude-aware operators succeed (\S\ref{subsec:edge-partitioned}).} Treating the strong head and the weak tail differently, using a regime boundary located by the diagnosis rather than by tuning, recovers the predicted ordering.
\end{itemize}
The neighbor cutoff $k$ on these graphs is therefore a \emph{semantic switch}, not a sparsification hyperparameter, and we turn the diagnosis into a reusable protocol practitioners can run before deployment (\S\ref{subsec:diagnostic-protocol}).

\paragraph{Scope of the claim.}
We do not claim choice modeling is unhelpful, nor that no choice-aware operator can outperform co-click. We claim that the natural \emph{drop-in} embodiment is structurally misaligned with within-slate ranking, and that this misalignment is diagnosable in advance. Our absolute improvements over co-click are modest, a property the diagnosis itself predicts under the attenuated propagation channel we observe (\S\ref{subsec:mixer-collapse}). The evidence for our claim is the \emph{pattern} of which interventions fail or succeed, not the size of any single headline gain.

\paragraph{Contribution.}
\begin{itemize}[leftmargin=*,topsep=2pt,itemsep=1pt]
  \item We \textbf{identify and characterize} an \emph{edge spectrum} on choice-derived item graphs, on which strong and weak edges encode qualitatively different relations (\S\ref{sec:why}).
  \item We \textbf{formalize the mechanism} as a sign mismatch (Prop.~\ref{prop:sign}), \textbf{prove} that co-click cannot satisfy the same condition (Lem.~\ref{lemma:coclick}), and confirm both predictions empirically.
  \item We provide the \textbf{first controlled comparison} of four item-side propagation operators on MIND and EB-NeRD, showing drop-in choice-derived operators do not outperform co-click despite indexing structurally distinct neighborhoods (\S\ref{sec:main-comparison}).
  \item We develop a \textbf{reusable diagnostic protocol} (\S\ref{subsec:diagnostic-protocol}) and pre-register an \textbf{edge-magnitude-aware operator} whose regime boundary is located by the diagnosis rather than by tuning, matching its predictions on both datasets (\S\ref{subsec:edge-partitioned}).
\end{itemize}

\section{Related Work}
\label{sec:related}

\subsection{Graph CF and item-side smoothing}
Graph collaborative filtering (CF) for implicit feedback originates with NGCF~\cite{model/wang2019neural} and was simplified by LightGCN~\cite{model/he2020lightgcn}; subsequent work refines layer combination~\cite{model/lightgcnpp}, approximates infinite-layer propagation~\cite{model/mao2021ultragcn}, adds contrastive perturbation~\cite{model/yu2022simgcl}, or reformulates propagation as a spectral filter~\cite{model/shen2021gfcf}. Across this family, item-side propagation is positive smoothing over a co-occurrence-derived item graph, either explicitly~\cite{model/sarwar2001item} or implicitly through $S^\top S$. Work that varies item-graph construction substitutes the \emph{source} of edges (knowledge graphs (KGs)~\cite{model/wang2019kgat,model/wang2019kgnnls}, content~\cite{model/he2016vbpr,model/zhang2016cke}, social ties~\cite{model/fan2019graph}) while preserving positive smoothing. Signed graph neural networks (GNNs)~\cite{model/derr2018signed,model/huang2019signed} use explicit edge signs, but only in settings where the sign is directly observable (e.g., trust networks). Our contribution within this thread is to identify a class of item edges, those derived from choice models, on which the appropriate sign is \emph{not} observable but is determined by interaction with the ranking objective.

\subsection{Choice modeling in recommendation: the expectation we test}
The multinomial logit (MNL) choice model~\cite{theory/mcfadden1973conditional} has a long history in econometrics for measuring substitution via diversion ratios~\cite{book/train2009discrete,theory/berry1994estimating}, and the principle that recommendation should account for local competition dates at least to Collaborative Competitive Filtering~\cite{model/ccf}. Within this line, MNL and its extensions have been used as ranking models~\cite{method/aouad2023assortment}, slate ranking losses~\cite{method/pang2020setrank}, and components in choice-aware sequential models~\cite{model/kang2018self,model/lcm4rec}. The operative position across these works is that choice-aware item relations encode a richer signal than co-occurrence ones; the natural translation into graph CF, namely using a choice-derived item--item graph as the propagation source, has not been directly compared against co-occurrence alternatives under controlled training, and we provide that comparison. Two Is Better Than One~\cite{model/kvernadze2022two} distinguishes ``similar'' from ``complementary'' product relations via separate embeddings; our distinction operates one level earlier, on the propagation graph itself, and is driven by the interaction between graph weights and the ranking gradient rather than by an a priori relation taxonomy.

\subsection{Rigorous re-evaluation and impression-aware recommendation}
Our methodology follows works that re-evaluate established methods under controlled protocols: Dacrema et al.~\cite{analysis/dacrema2019are}, Rendle et al.~\cite{analysis/rendle2020neural}, 
Sato et al.~\cite{analysis/sato2022reevaluating}, and Iana et al.~\cite{analysis/iana2024simplifying}. We apply the same frame to choice-derived item-side operators: freeze the baseline pack before evaluating the proposed construction, and trace any negative result to a structural source. News recommendation provides the primary empirical testbed for impression-aware learning, with MIND~\cite{dataset/wu2020mind} and EB-NeRD~\cite{dataset/kruse2024ebnerd} as principal benchmarks supporting session-aware methods~\cite{model/wu2019nrms,model/an2019lstur,model/wu2019naml,model/li2022miner,model/qi2022fum}. Displayed-but-unclicked items are variously used as in-batch negatives~\cite{model/wu2019npa,model/wu2019nrms}; slate-aware rankers~\cite{method/sunehag2015slate,method/jiang2019listcvae} and MNL slate-choice rankers~\cite{method/pang2020setrank,method/aouad2023assortment} model within-slate competition at the level of the loss or policy, but, to our knowledge, no prior work has tested whether the item--item edges used for graph propagation should themselves be derived from the slate-level choice context.

\section{Setup: Item-Side Operators in Graph CF}
\label{sec:setup}

\noindent\textit{\textbf{Take-away.} We isolate the item-graph axis by swapping LightGCN's implicit operator $S^\top S$ for one of four candidates, holding all else fixed.}

\vspace{0.5em}
\noindent We take LightGCN's propagation polynomial, identify the item-side even-hop operator within it, and swap that operator for one of four candidates. All other components (bipartite structure, layer combination, scoring rule, training objective, and hyperparameters) are held fixed. Table~\ref{tab:notation} summarizes the notation used throughout the paper.

\begin{table}[t]
\caption{Notation used throughout the paper.}
\vspace{-0.2cm}
\label{tab:notation}
\small
\setlength{\tabcolsep}{4pt}
\begin{tabular}{@{}lp{0.66\columnwidth}@{}}
\toprule
\textbf{Symbol} & \textbf{Meaning} \\
\midrule
\multicolumn{2}{@{}l}{\emph{Data and impressions}} \\
$\mathcal{U}, \mathcal{I}$       & user set, item set \\
$t$                              & an impression (a single user--slate interaction) \\
$C_t$                            & candidate slate at impression $t$ \\
$Y_t \subseteq C_t$              & set of clicked items at impression $t$ \\
$u_t$                            & user at impression $t$ \\
\midrule
\multicolumn{2}{@{}l}{\emph{Graph operators}} \\
$S$                              & symmetrically-normalized user--item interaction matrix \\
$C_I$                            & item-side propagation operator (one of four candidates) \\
$S^\top S$                       & LightGCN's implicit item-side operator \\
$W_{ij}^{\mathrm{cc}}$           & co-click weight: \# users who clicked both $i$ and $j$ \\
$W_{ij}^{\mathrm{coexp}}$        & co-exposure weight: \# impressions where $i, j$ co-appear \\
$W_{ij}^{\mathrm{sub}}$          & symmetrized MNL diversion weight (choice-derived) \\
\midrule
\multicolumn{2}{@{}l}{\emph{Choice model}} \\
$P(i \mid u, C)$                 & MNL probability of selecting $i$ from slate $C$ \\
$\Delta_{j \to i}(u, C)$         & MNL diversion weight from $j$ to $i$ (Eq.~\ref{eq:diversion}) \\
$\Wtilde_{ij}$                    & symmetrized diversion weight before pruning \\
\midrule
\multicolumn{2}{@{}l}{\emph{Edge spectrum and diagnostics}} \\
strong edges                     & high-weight head of $W^{\mathrm{sub}}$; concentrated on in-slate competitors \\
weak edges                       & low-weight tail of $W^{\mathrm{sub}}$; not concentrated on competitors \\
$N_o^{(k)}(i)$                   & top-$k$ neighbors of item $i$ under operator $o$ \\
$\Delta_o(k)$                    & within-slate lift at cutoff $k$ for operator $o$ (Eq.~\ref{eq:lift}) \\
$K_{\mathrm{strong}}$            & cutoff that separates strong from weak edges \\
\bottomrule
\end{tabular}
\vspace{-0.5cm}
\end{table}

\subsection{LightGCN's Implicit Item-Item Operator and Four Candidates}
\label{subsec:lightgcn-implicit}

Let $R \in \mathbb{R}^{|\mathcal{U}| \times |\mathcal{I}|}$ denote the user--item interaction matrix derived from impressions, with entries
\begin{equation}
R_{ui} = \log(1 + \mathrm{expose}_{ui}) \cdot \frac{\mathrm{click}_{ui} + a}{\mathrm{expose}_{ui} + a + b}
\label{eq:Rui}
\end{equation}
for fixed smoothing constants $a, b > 0$. The first factor handles the heavy-tailed exposure distribution and the second is a Beta-$(a, b)$-smoothed click rate; we use $a = 1$, $b = 10$ throughout (prior click probability $a/(a+b) \approx 9\%$, matching the empirical news CTR regime). Symmetric degree normalization gives $S = D_U^{-1/2} R\, D_I^{-1/2}$ and the bipartite operator $A = \bigl(\begin{smallmatrix} 0 & S \\ S^\top & 0 \end{smallmatrix}\bigr)$. When the propagation polynomial is expanded, item-side even-hop terms take the form $(S^\top S)^k$: the matrix $S^\top S$ is itself a normalized item--item graph and is the implicit item--item operator in LightGCN-style propagation. We parameterize item-side propagation as $(C_I)^k$ for a symmetric, non-negative, degree-normalized, top-$K$-pruned operator $C_I$; substituting $C_I = S^\top S$ recovers LightGCN's polynomial at the operator level (Appendix~\ref{app:parity}). It does \emph{not} entail numerical equivalence with the locked LightGCN baseline because the training-time framework differs; Section~\ref{subsec:finding3-framework-gap} reports a framework-axis ablation. Our setup therefore enables apples-to-apples comparison \emph{across operator choices}, not against absolute LightGCN-family performance; the former is the comparison the spectrum analysis requires. We use uniform layer-combination weights throughout; a learnable softmax mixer collapses onto the un-propagated user term (Section~\ref{subsec:mixer-collapse}).

We compare four constructions of $C_I$ in order of increasing choice-theoretic content. \emph{$S^\top S$} is LightGCN's implicit operator. \emph{Co-click} is the explicit item--item graph $W_{ij}^{\mathrm{cc}} = \#\{u : \mathrm{click}_{ui} = \mathrm{click}_{uj} = 1\}$. \emph{Co-exposure} has $W_{ij}^{\mathrm{coexp}} = \#\{t : i, j \in C_t\}$, encoding platform-level grouping but ignoring within-slate choice. \emph{Choice-derived substitution} is built from MNL diversion weights (Section~\ref{subsec:substitution-build}) and is the operator the choice-modeling literature~\cite{model/lcm4rec,model/ccf,method/aouad2023assortment} treats as a richer alternative to co-occurrence. All four undergo identical symmetric degree normalization and top-$K$ pruning ($K = 50$) and plug into the same item-side slot without any other code change.

\subsection{Building the Substitution Operator}
\label{subsec:substitution-build}

We fit a factorized MNL choice model on impression data with utility $v(u, i) = p_u^\top q_i + b_i$ and slate-level choice probability $P(i \mid u, C) = \exp v(u, i) / \sum_{j \in C} \exp v(u, j)$, by minimizing slate-level cross-entropy (held-out negative log-likelihood, NLL: $3.22$ on MIND, $1.92$ on EB-NeRD). The MNL diversion weights are
\begin{equation}
\Delta_{j \to i}(u, C) = \frac{P(i \mid u, C)\, P(j \mid u, C)}{1 - P(j \mid u, C)}.
\label{eq:diversion}
\end{equation}
We aggregate across impressions in which $i, j$ co-appear, symmetrize $\Wtilde_{ij} = \tfrac{1}{2}(\Delta_{j \to i} + \Delta_{i \to j})$, and apply top-$K$ pruning and symmetric degree normalization to obtain $\Csub = D^{-1/2} \Wtilde_{\mathrm{top}\text{-}K} D^{-1/2}$ (construction statistics in Table~\ref{tab:graph-stats}). The MNL model is used only to construct $\Csub$; its embeddings are not shared with the propagation framework.

\subsection{Controlled Comparison Protocol}
\label{subsec:protocol}

We freeze three components before any operator variant is run, following~\cite{analysis/dacrema2019are,analysis/rendle2020neural}: (i) data and impression-level splits; (ii) the evaluator computing normalized discounted cumulative gain at 10 (NDCG@10), mean reciprocal rank (MRR), Hit@1, and group AUC (gAUC) on the official validation split; (iii) the baseline pack (LightGCN, LightGCN++, BPR-MF, GFCF, ItemKNN variants, MNL-only slate-choice ranker). After the lock, no baseline hyperparameter is altered. Training minimizes within-impression Bayesian personalized ranking (BPR) loss
\begin{equation}
\mathcal{L}_{t,ij} = -\log \sigma\bigl(s(u_t, i) - s(u_t, j)\bigr)
\label{eq:bpr}
\end{equation}
with in-slate negatives, identical across all four operator variants and three random seeds, where $s(u_t, i) = \langle \Euserhat[u_t], \Eitemhat[i] \rangle$. Code and trained MNL choice models will be released upon acceptance; baseline settings are in Appendices~\ref{app:mixer}--\ref{app:remedies}.

\section{Drop-In Operators: A Narrow Band on the Headline Metric}
\label{sec:main-comparison}
\noindent\textit{\textbf{Take-away.} Under controlled training, the four item-side operators land within a narrow band ($0.005$--$0.011$ NDCG@10); the drop-in choice-derived operator does not beat co-click.}

\vspace{0.5em}
\noindent This section establishes three properties about drop-in operator swaps: (i) the four operators land within a narrow band despite indexing structurally different relations; (ii) the drop-in choice-derived substitution operator does not outperform co-click; (iii) the propagation channel is substantially attenuated in our setup, which constrains the absolute scale of any operator-level intervention but does not invalidate the structural analysis of \S\ref{sec:why}. Framework-axis robustness is addressed in Section~\ref{subsec:finding3-framework-gap} and Appendix~\ref{app:framework-axis}.

\subsection{Finding 1: Operator Choice Matters, But Within a Narrow Band}
\label{subsec:finding1-narrow-band}

Table~\ref{tab:operator-ablation} reports NDCG@10 on the internal development split over three seeds; Table~\ref{tab:main-results} (Section~\ref{subsec:finding3-framework-gap}) reports the official validation split used for baseline comparison. Operator-relative orderings are consistent across the two splits in the sense that substitution does not strictly outperform co-click in any of the four comparisons (two splits $\times$ two datasets); absolute values differ as expected between splits. On MIND the four operators span $0.4036$--$0.4088$ (spread $0.0052$); on EB-NeRD the spread is $0.0083$. The spread is small in absolute magnitude (on the order of seed standard deviation), but the ordering is consistent across seeds within each dataset and across datasets at the top. The narrow band itself is worth flagging: item-side operator construction, which has been the focus of substantial graph CF design effort (knowledge-graph augmentation, signed graphs, self-supervised constructions), moves headline NDCG@10 by less than $0.01$ on these datasets, holding everything else constant. Where it does move it, the simpler operator wins.

\begin{table}[t]
\caption{Operator ablation under uniform mixer (internal dev split, mean $\pm$ std over 3 seeds, NDCG@10). \textbf{Co-click is the strongest operator on both datasets}, despite substitution being constructed from a richer choice-theoretic signal. The four operators span a narrow band ($\leq 0.01$ NDCG@10), with substitution placing at or near the bottom: the empirical anchor for the structural analysis of \S\ref{sec:why}.}
\label{tab:operator-ablation}
\small
\setlength{\tabcolsep}{6pt}
\begin{tabular}{lcc}
\toprule
\textbf{Operator} & \textbf{MIND} & \textbf{EB-NeRD} \\
\midrule
co-click           & $\mathbf{0.4088 \pm 0.0011}$ & $\mathbf{0.5985 \pm 0.0010}$ \\
$S^\top S$         & $0.4052 \pm 0.0006$ & $0.5953 \pm 0.0021$ \\
co-exposure        & $0.4054 \pm 0.0008$ & $0.5902 \pm 0.0011$ \\
substitution       & $0.4036 \pm 0.0019$ & $0.5905 \pm 0.0009$ \\
\midrule
spread (best$-$worst) & $0.0052$ & $0.0083$ \\
\bottomrule
\end{tabular}
\vspace{-0.3cm}
\end{table}

\subsection{Finding 2: The Drop-In Choice-Derived Operator Does Not Outperform Co-Click}
\label{subsec:finding2-coclick-suffices}

Co-click is the best operator on both datasets on the internal development split. The substitution operator, constructed from a richer choice-theoretic signal by fitting an MNL choice model and extracting diversion weights, places at or near the bottom. The pattern is reproducible at the dev-split level: across three seeds, co-click beats substitution by $0.0029$--$0.0067$ NDCG@10 on MIND (mean $0.0052$) and by $0.0071$--$0.0089$ on EB-NeRD (mean $0.0080$). On the official validation split (Table~\ref{tab:main-results}), substitution and co-click are statistically indistinguishable on EB-NeRD ($0.5984$ vs.\ $0.5971$, within seed standard deviation) while co-click leads on MIND; across the four comparisons (two splits $\times$ two datasets), substitution does not strictly outperform co-click in any. \emph{The richer signal does not translate into improved drop-in propagation.} This is the empirical anchor for the structural analysis of Section~\ref{sec:why}: the analysis must explain not only why substitution fails to surpass co-click but also why the failure persists despite the two graphs indexing largely disjoint neighborhoods (Section~\ref{subsec:graph-distinct}).

\subsection{The Two Graphs Are Structurally Distinct}
\label{subsec:graph-distinct}

A skeptical reader might ask whether the substitution graph inadvertently rediscovers co-click in elaborated form, in which case Finding~2 would be a trivial consequence of redundancy. Table~\ref{tab:graph-stats} reports overlap between the substitution graph and each of the three co-occurrence-based operators. Against co-click, mean Jaccard@50 is $0.014$ on MIND and $0.044$ on EB-NeRD; weighted cosine is $0.026$ and $0.139$. The substitution graph and the co-click graph index largely disjoint neighborhoods. Against $S^\top S$ the overlap is even smaller. The partial exception is co-exposure (Jaccard@50 $0.20$--$0.22$), reflecting that both graphs draw signal from items that appear together in slates; but substitution weights co-appearances by their MNL choice probabilities while co-exposure weights them equally, so even this overlap is well below identity. Finding~2 is therefore unlikely to be a rediscovery artifact: the two graphs index different relations, yet co-click is the better propagation source. The structural reason is the subject of Section~\ref{sec:why}.

\begin{table}[t]
\caption{Structural overlap between the substitution graph and three co-occurrence-based operators. \textbf{The substitution graph indexes a largely disjoint neighborhood from both co-click and $S^\top S$} (Jaccard@50 $\leq 0.05$ on both datasets), confirming that the failure of substitution in Table~\ref{tab:operator-ablation} is not a rediscovery artifact. The partial exception is co-exposure (Jaccard@50 $\approx 0.2$), since both graphs draw signal from in-slate co-appearance, but substitution weights these by MNL choice probabilities while co-exposure weights equally.}
\label{tab:graph-stats}
\small
\setlength{\tabcolsep}{6pt}
\begin{tabular}{lcc}
\toprule
\textbf{Compared to substitution graph} & \textbf{MIND} & \textbf{EB-NeRD} \\
\midrule
\multicolumn{3}{l}{\emph{vs.\ co-click}} \\
Mean Jaccard@50               & 0.014 & 0.044 \\
Weighted cosine               & 0.026 & 0.139 \\
\midrule
\multicolumn{3}{l}{\emph{vs.\ $S^\top S$}} \\
Mean Jaccard@50               & 0.014 & 0.053 \\
Weighted cosine               & 0.009 & 0.062 \\
\midrule
\multicolumn{3}{l}{\emph{vs.\ co-exposure}} \\
Mean Jaccard@50               & 0.203 & 0.216 \\
Weighted cosine               & 0.318 & 0.448 \\
\midrule
\multicolumn{3}{l}{\emph{Substitution graph construction (reference)}} \\
Choice model dev NLL          & 3.22  & 1.92  \\
Retained edges (top-$K$ pruned)& $1{,}268{,}104$ & $160{,}968$ \\
\bottomrule
\end{tabular}
\vspace{-0.3cm}
\end{table}

\subsection{Finding 3: The Gap to Locked Baselines Is Outside the Item-Side Operator Axis}
\label{subsec:finding3-framework-gap}

Table~\ref{tab:main-results} shows the full baseline pack against the four operator variants. The best variant trails LightGCN++ by $0.083$ NDCG@10 on MIND and $0.054$ on EB-NeRD, an order of magnitude larger than the operator spread ($0.005$--$0.008$). Decomposing this gap (Appendix~\ref{app:framework-axis}, Table~\ref{tab:framework-ablation}), the LightGCN-aligned axes (binary $R$, no pruning) jointly explain less than $0.002$ NDCG@10; the residual $\approx 0.04$ is attributable to training-time hyperparameter and architecture choices outside the operator axis. \emph{The improvement LightGCN++ delivers over LightGCN is therefore not located in item-graph construction}; we return to this in Section~\ref{sec:discussion}.

\begin{table*}[t]
\caption{Validation-split comparison. \textbf{Bold} marks the overall best per dataset; \underline{underline} marks the best within ``Ours''. The four ``Ours'' variants span a narrow band ($\leq 0.01$ NDCG@10), with substitution never strictly outperforming co-click; the gap to LightGCN++ is an order of magnitude larger and lies outside the item-side operator axis (Section~\ref{subsec:finding3-framework-gap}, Appendix~\ref{app:framework-axis}).}
\label{tab:main-results}
\small
\setlength{\tabcolsep}{4pt}
\begin{tabular}{llcccccccc}
\toprule
& & \multicolumn{4}{c}{\textbf{MIND}} & \multicolumn{4}{c}{\textbf{EB-NeRD}} \\
\cmidrule(lr){3-6}\cmidrule(lr){7-10}
\textbf{Group} & \textbf{Method} & NDCG@10 & MRR & Hit@1 & gAUC & NDCG@10 & MRR & Hit@1 & gAUC \\
\midrule
Sanity            & MostPop                                                          & 0.3918 & 0.3304 & 0.1570 & 0.4977 & 0.6353 & 0.5316 & 0.3150 & 0.6339 \\
Sanity            & Random                                                           & 0.3792 & 0.3164 & 0.1432 & 0.4946 & 0.5463 & 0.4336 & 0.2038 & 0.4950 \\
\addlinespace[2pt]
Non-graph CF      & BPR-MF                                                           & 0.4575 & 0.3872 & 0.2076 & 0.6042 & 0.6321 & 0.5317 & 0.3204 & 0.6194 \\
\addlinespace[2pt]
Graph CF          & GFCF                                                             & 0.3910 & 0.3253 & 0.1480 & 0.4887 & 0.5974 & 0.4867 & 0.2591 & 0.5806 \\
Graph CF          & LightGCN                                                         & 0.4344 & 0.3675 & 0.1863 & 0.5663 & 0.6358 & 0.5357 & 0.3200 & 0.6246 \\
Graph CF          & LightGCN++                                                       & \textbf{0.4754} & \textbf{0.4083} & \textbf{0.2280} & \textbf{0.6174} & \textbf{0.6523} & \textbf{0.5545} & \textbf{0.3454} & \textbf{0.6538} \\
\addlinespace[2pt]
Co-occurrence     & ItemKNN-CoClick                                                  & 0.3907 & 0.3271 & 0.1440 & 0.5053 & 0.5957 & 0.4874 & 0.2606 & 0.5710 \\
Co-occurrence     & ItemKNN-CoExposure                                               & 0.3668 & 0.2997 & 0.1243 & 0.4870 & 0.5452 & 0.4297 & 0.1968 & 0.4977 \\
\addlinespace[2pt]
Choice-only       & MNL-Slate                                                        & 0.4575 & 0.3888 & 0.2092 & 0.6074 & 0.6297 & 0.5301 & 0.3188 & 0.6156 \\
\specialrule{0.7pt}{2pt}{2pt}
Ours              & co-click operator                                                & \underline{0.3926} & \underline{0.3228} & \underline{0.1426} & \underline{0.5291} & 0.5971 & 0.4933 & 0.2734 & \underline{0.6005} \\
Ours              & co-exposure operator                                             & 0.3921 & 0.3214 & 0.1410 & 0.5251 & 0.5964 & 0.4932 & 0.2768 & 0.5987 \\
Ours              & substitution operator                                            & 0.3857 & 0.3147 & 0.1334 & 0.5219 & \underline{0.5984} & \underline{0.4957} & \underline{0.2796} & 0.6002 \\
Ours              & $S^\top S$ operator                                              & 0.3917 & 0.3216 & 0.1424 & 0.5249 & 0.5962 & 0.4920 & 0.2723 & 0.5978 \\
\bottomrule
\end{tabular}
\vspace{-0.2cm}
\end{table*}

\subsection{Finding 4: The Operator Ordering Is Robust to Framework Knobs}
\label{subsec:finding4-lightgcn-aligned}

Re-running under an alternative LightGCN-compatible framework (binary $R$, random-negative BPR, common top-$K{=}50$ pruning), Table~\ref{tab:operator-ablation-lightgcn-aligned} shows the hierarchy from Table~\ref{tab:operator-ablation} is preserved: \{co-click, $S^\top S$\} statistically indistinguishable at the top, substitution and co-exposure below. The narrow band ($\leq 0.012$) and substitution's inability to surpass co-click persist under framework changes.

\begin{table}[t]
\caption{Four-operator comparison under an alternative LightGCN-aligned framework (binary $R$, random-negative BPR, common top-$K{=}50$ pruning), NDCG@10 on validation. \textbf{The hierarchy from Table~\ref{tab:operator-ablation} is preserved}: \{co-click, $S^\top S$\} statistically indistinguishable at the top, substitution and co-exposure below. The narrow band ($\leq 0.012$) and substitution's inability to surpass co-click persist under framework changes. $S^\top S$ and co-click use 3 seeds; others use seed=1.}
\label{tab:operator-ablation-lightgcn-aligned}
\small
\setlength{\tabcolsep}{6pt}
\begin{tabular}{lcc}
\toprule
\textbf{Operator} & \textbf{MIND} & \textbf{EB-NeRD} \\
\midrule
co-click           & $0.3744 \pm 0.0107$ & $\mathbf{0.5283 \pm 0.0024}$ \\
$S^\top S$         & $\mathbf{0.3750 \pm 0.0109}$ & $0.5276 \pm 0.0015$ \\
substitution       & $0.3653$            & $0.5200$ \\
co-exposure        & $0.3636$            & $0.5189$ \\
\midrule
spread (best$-$worst) & $0.0114$ & $0.0094$ \\
\bottomrule
\end{tabular}
\vspace{-0.3cm}
\end{table}

\subsection{Why the Band Is So Narrow: A Mixer Collapse Observation}
\label{subsec:mixer-collapse}

Why is the operator-variation band so narrow ($0.005$--$0.008$ NDCG@10) given that the four operators index structurally different relations? The narrow band reflects, in part, a property of the propagation channel itself in our training setup. \emph{This observation constrains the absolute scale of operator-level interventions but does not undermine the spectrum analysis, for reasons we make explicit below.}

We replace the uniform layer-combination weights with a learnable softmax mixer $\alpha_k = \exp(\theta_k) / \sum_{k'} \exp(\theta_{k'})$ trained jointly with the embeddings. Under within-impression BPR, the mixer collapses onto the un-propagated user term: even-parity mass exceeds $0.9996$ on both datasets with normalized entropy below $0.005$ (Table~\ref{tab:mixer-collapse}, Appendix~\ref{app:mixer}). Imposing a floor $\alpha_k \geq \alpha_{\min} \in \{0.25, 0.5\}$ preventing collapse modestly improves dev NDCG@10 ($+0.009$ on MIND, $+0.006$ on EB-NeRD) but does not exceed the uniform-coefficient configuration, and the operator ordering of Table~\ref{tab:operator-ablation} is unchanged.

\paragraph{Scope of the observation.} Mixer collapse explains the \emph{magnitude} of operator-level effects but not their \emph{direction}. The spectrum is a property of the static graph and the trained scoring function: the within-slate lift $\Delta_o(k)$ (Eq.~\ref{eq:lift}) measures which item pairs the trained embeddings score similarly, independent of the propagation channel's magnitude. The diagnosis therefore predicts a qualitative \emph{ordering}: positive smoothing on substitution (incorrect on strong edges) $<$ uniform sign flip (incorrect across both regimes) $<$ edge-magnitude-aware (matches graph structure); we observe exactly this ordering in Tables~\ref{tab:edge-partitioned} and~\ref{tab:remedies-appendix}. The propagation channel is attenuated, not nullified: imposing a floor improves NDCG@10 and preserves the operator ordering of Table~\ref{tab:operator-ablation}.

\section{The Edge Spectrum}
\label{sec:why}

\noindent\textit{\textbf{Take-away.} A sign mismatch (Proposition~\ref{prop:sign}) connects substitution edges to within-slate BPR; co-click cannot satisfy the same condition (Lemma~\ref{lemma:coclick}). The lift profile (Figure~\ref{fig:kspectrum}) confirms both predictions, with a sign change across $k$ on EB-NeRD.}

\vspace{0.5em}
\noindent This section identifies the structural reason behind Section~\ref{sec:main-comparison}'s observation. Choice-derived item graphs exhibit an \emph{edge spectrum}: their strong and weak edges encode qualitatively different relations, and this property is absent on co-click by construction. We formalize the mechanism (Proposition~\ref{prop:sign}, Lemma~\ref{lemma:coclick}) and confirm both predictions empirically (Section~\ref{subsec:lift-confirmation}).

\subsection{Setup and Notation}
\label{subsec:why-setup}

Fix an impression $t$ with user $u_t$, candidate slate $C_t$, and click set $Y_t \subseteq C_t$. For any item-side operator $C_I$ with symmetric non-negative weights $W_{k\ell}$, consider a single first-order smoothing step:
\begin{equation}
\tilde{x}_k = x_k + \alpha \sum_{\ell} W_{k\ell}\, x_\ell, \qquad \alpha > 0,
\label{eq:smoothing-step}
\end{equation}
followed by bilinear scoring $s_k(u) = \langle h_u, x_k \rangle$. The within-slate ranking objective for any clicked--unclicked pair $(i,j) \in Y_t \times (C_t \setminus Y_t)$ is $\mathcal{L}_{ij} = -\log \sigma(s_i - s_j)$. We use $W_{ij}^{\mathrm{sub}}$ for the symmetrized MNL diversion weight and $W_{ij}^{\mathrm{cc}}$ for the symmetric co-click weight, both after degree normalization.

\subsection{Sign Mismatch for Choice-Derived Edges}
\label{subsec:proposition}

The MNL diversion weight is bounded below by $\Delta_{j \to i}(u, C) \geq P(i \mid u, C) P(j \mid u, C)$, so $W_{ij}^{\mathrm{sub}}$ aggregates a quantity that is large precisely when both $i$ and $j$ are jointly plausible choices in the same slate. This is what makes strong substitution edges concentrate on in-slate clicked--unclicked pairs.

\begin{proposition}[Sign mismatch on choice-derived edges]
\label{prop:sign}
Let $(i, j)$ be a clicked--unclicked pair in impression $t$ such that $W_{ij}^{\mathrm{sub}}$ is in the top-$K$ row of $\Csub$ for $i$, with $W_{ij}^{\mathrm{sub}} \geq \sum_{\ell \neq j} W_{i\ell}^{\mathrm{sub}}$ (top-$K$-dominance). Under the smoothing step of Eq.~\ref{eq:smoothing-step}, the within-slate BPR loss $\mathcal{L}_{ij}$, and bilinear scoring, the score gap evolves to first order in $\alpha$ as
\begin{equation}
\tilde{s}_i - \tilde{s}_j = (1 - \alpha W_{ij}^{\mathrm{sub}})\, (s_i - s_j) + O(\alpha \cdot \epsilon_{ij}),
\label{eq:score-gap-evolution}
\end{equation}
where $\epsilon_{ij} := \max(\sum_{\ell \neq j} W_{i\ell}^{\mathrm{sub}}, \sum_{\ell \neq i} W_{j\ell}^{\mathrm{sub}}) \leq W_{ij}^{\mathrm{sub}}$ under top-$K$-dominance. Meanwhile $\nabla_\theta \mathcal{L}_{ij}$ pushes $s_i - s_j$ upward. The substitution operator and the ranking gradient therefore act on $(s_i, s_j)$ in opposing directions whenever $\alpha W_{ij}^{\mathrm{sub}} > 0$.
\end{proposition}

\begin{proof}
By symmetry of $W^{\mathrm{sub}}$, $\tilde{x}_i - \tilde{x}_j = (x_i - x_j) - \alpha W_{ij}^{\mathrm{sub}} (x_i - x_j) + R_{ij}$, where $R_{ij} = \alpha \sum_{\ell \neq i,j} (W_{i\ell}^{\mathrm{sub}} - W_{j\ell}^{\mathrm{sub}}) x_\ell$. Under top-$K$-dominance, $\|R_{ij}\| = O(\alpha \epsilon_{ij})$. Applying bilinear scoring yields Eq.~\ref{eq:score-gap-evolution}. The gradient $\nabla_{(s_i - s_j)} \mathcal{L}_{ij} = -\sigma(-(s_i - s_j)) < 0$, so a gradient step strictly increases $s_i - s_j$.
\end{proof}

\paragraph{What this means.}
Smoothing reduces the score gap between a clicked item and its strong-edge competitor by a factor $(1 - \alpha W_{ij}^{\mathrm{sub}})$. The ranking gradient simultaneously increases it. The two operations fight each other on exactly the pairs where the substitution weight is largest, namely the in-slate competitors that the choice model identifies as competing for the same click. Three facts compose this result: (i) symmetric positive smoothing shrinks score gaps; (ii) the within-slate BPR gradient grows them; and (iii) \emph{the MNL diversion form makes top-$K$-dominance hold precisely on within-slate clicked--unclicked pairs}. Facts (i) and (ii) hold for any non-negative graph; fact (iii) is what distinguishes choice-derived graphs from co-click and is what Lemma~\ref{lemma:coclick} shows fails by construction for co-click.

\paragraph{Remark (scope of the proposition).}
Proposition~\ref{prop:sign} states the \emph{direction} of effect on top-$K$-dominant pairs to first order in $\alpha$; cross-slate signal and optimizer dynamics modulate the net training equilibrium, which the lift profile of Figure~\ref{fig:kspectrum} captures empirically. Top-$K$-dominance naturally emerges from MNL diversion because slate-level choice probabilities concentrate on a few high-utility items per impression. Beyond this regime, $\epsilon_{ij}$ becomes comparable to $W_{ij}^{\mathrm{sub}}$ and the first-order direction is no longer predicted; the lift profile accordingly decays toward zero or reverses sign past the regime boundary, located by the diagnostic of \S\ref{subsec:diagnostic-protocol} at $k \approx 10$ and respected by the edge-partitioned operator of \S\ref{subsec:edge-partitioned}.

\subsection{Why Co-Click Does Not Satisfy the Same Misalignment}
\label{subsec:lemma-coclick}

The structural difference between the two graphs is at the level of \emph{which pairs} carry large weight, not at the level of the smoothing operator itself.

\begin{lemma}[Co-click does not concentrate on in-slate pairs]
\label{lemma:coclick}
For any impression $t$ with $i \in Y_t$ and $j \notin Y_t$, the contribution of impression $t$ to $W_{ij}^{\mathrm{cc}}$ is zero. Consequently, $W_{ij}^{\mathrm{cc}}$ is large only when $i$ and $j$ have been \emph{co-positive} across the user history, a relation structurally distinct from the within-slate clicked--unclicked pair structure that Proposition~\ref{prop:sign} reasons about: no impression $t$ with $i \in Y_t$ and $j \notin Y_t$ contributes to $W_{ij}^{\mathrm{cc}}$, while exactly such impressions dominate the construction of $W_{ij}^{\mathrm{sub}}$.
\end{lemma}

\begin{proof}
By construction $W_{ij}^{\mathrm{cc}} = \sum_u \mathbb{1}[\mathrm{click}_{ui} = 1 \wedge \mathrm{click}_{uj} = 1]$. For impression $t$ with $i \in Y_t$, $j \notin Y_t$, $\mathrm{click}_{u_t, j} = 0$, so the indicator is zero. Large $W_{ij}^{\mathrm{cc}}$ requires that across other impressions or users both $i$ and $j$ were clicked. The within-slate structure measures co-candidacy with disagreement; the co-click weight measures co-positivity. The two sources of weight are constructed from disjoint event types.
\end{proof}

\paragraph{What this means.}
Proposition~\ref{prop:sign} requires $W_{ij}$ be large on the same pairs the ranking gradient acts on. Choice-derived graphs satisfy this by construction; co-click graphs violate it by construction. The same propagation operator therefore induces different geometric effects on the two graphs.

\subsection{Empirical Confirmation: Within-Slate Lift Across Cutoffs}
\label{subsec:lift-confirmation}

Proposition~\ref{prop:sign} and Lemma~\ref{lemma:coclick} jointly predict (i) that the substitution operator should raise the scores of unclicked in-slate items connected to clicked items via high-weight edges, and (ii) that co-click should not exhibit a comparable effect. For each impression $t$ with at least one click and a clicked item $i \in Y_t$, partition the unclicked candidates $C_t \setminus Y_t$ into operator-neighbors of $i$ at cutoff $k$ and non-neighbors, and define
\begin{equation}
\Delta_o(k) = \mathbb{E}_{t, i \in Y_t}\!\left[
  \overline{s}\bigl(N_o^{(k)}(i) \cap (C_t \setminus Y_t)\bigr)
  - \overline{s}\bigl((C_t \setminus Y_t) \setminus N_o^{(k)}(i)\bigr)
\right],
\label{eq:lift}
\end{equation}
where $N_o^{(k)}(i)$ is the top-$k$ neighborhood and $\overline{s}(\cdot)$ is mean predicted score. Figure~\ref{fig:kspectrum} reports $\Delta_o(k)$ for all four operators across cutoffs.

\begin{figure*}[t]
  \centering
  \includegraphics[width=0.92\textwidth]{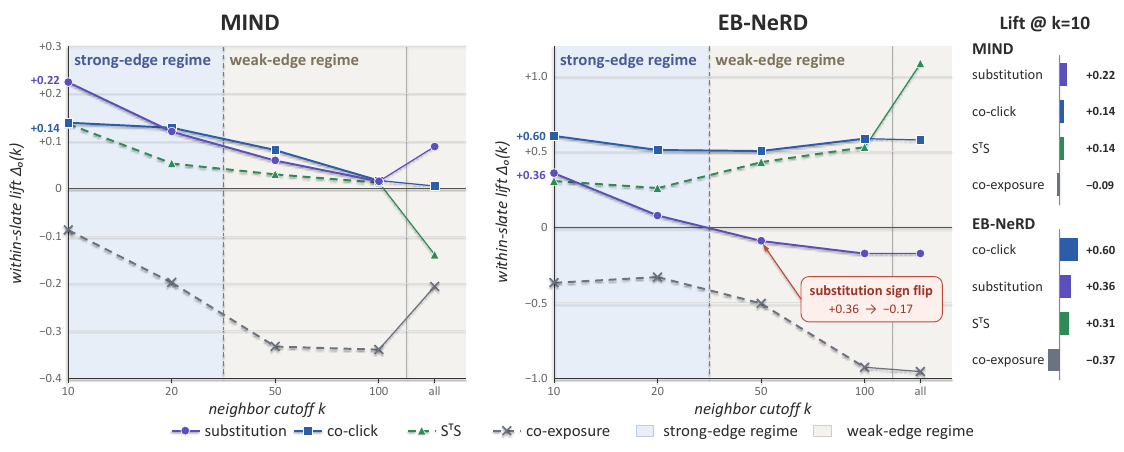}
  \vspace{-0.3cm}
  \caption{\textbf{The edge spectrum, visualized.} Within-slate lift $\Delta_o(k)$ for the four operators across neighbor cutoffs $k$. The $k=10$ region is the \textbf{strong-edge regime} (top-$K$-dominance holds, Proposition~\ref{prop:sign} applies); $k \geq 50$ is the \textbf{weak-edge regime} (top-$K$-dominance breaks down). On EB-NeRD, substitution flips sign between these regimes ($+0.36 \to -0.17$), confirming that two values of the same cutoff index relations of opposite sign. Co-click is flat throughout (Lemma~\ref{lemma:coclick}). Mean $\pm$ std over 3 seeds.}
  \label{fig:kspectrum}
  \Description{Line plots of the within-slate competitor lift as a function of the neighbor cutoff k for four item-side operators on the MIND and EB-NeRD datasets, showing that the substitution operator has large positive lift at small k and decays or flips sign at large k, while co-click stays nearly flat.}
  \vspace{-0.3cm}
\end{figure*}

The figure confirms both predictions. At small $k$, substitution lift is large and positive on both datasets: $\Delta_{\mathrm{sub}}(10) = +0.22$ on MIND (largest of the four) and $+0.36$ on EB-NeRD (second to co-click's $+0.60$). Strong substitution edges index items the trained model scores higher than other in-slate non-neighbors, as Proposition~\ref{prop:sign} predicts. As $k$ grows, this lift decays or reverses: on MIND $\Delta_{\mathrm{sub}}$ falls to $+0.015$ at $k=100$; on EB-NeRD it flips sign between $k=20$ and $k=50$ and remains negative through $k = $ all ($-0.17$). Beyond the strong-edge regime, $\epsilon_{ij}$ becomes comparable to $W_{ij}^{\mathrm{sub}}$ and the first-order approximation no longer predicts a single sign. By contrast $\Delta_{\text{co-click}}$ is nearly flat across $k$ (MIND $+0.14 \to +0.01$; EB-NeRD $+0.60 \to +0.58$): co-click neighborhoods do not partition into qualitatively distinct sub-populations by weight, as Lemma~\ref{lemma:coclick} predicts. Co-exposure produces large negative lift across all $k$ ($-0.09$ on MIND, $-0.37$ on EB-NeRD at $k=10$), indexing a third relation distinct from competition and shared preference.

\begin{table}[t]
\caption{Within-slate competitor lift at $k{=}10$ vs.\ headline NDCG@10 (mean $\pm$ std, 3 seeds). \textbf{On MIND, lift and NDCG anti-correlate across operators}: substitution attains the largest lift and the lowest NDCG, while co-click pairs a modest lift with the highest NDCG, as Proposition~\ref{prop:sign} predicts. On EB-NeRD, co-click leads in both, but substitution's misalignment does not surface as a headline penalty (analyzed in \S\ref{sec:when}).}
\label{tab:lift-vs-ndcg}
\small
\setlength{\tabcolsep}{6pt}
\begin{tabular}{llcc}
\toprule
\textbf{Dataset} & \textbf{Operator} & \textbf{Lift @ $k{=}10$} & \textbf{NDCG@10} \\
\midrule
\multirow{4}{*}{MIND}
        & co-click      & $+0.138 \pm 0.004$         & $\mathbf{0.4088 \pm 0.0011}$ \\
        & $S^\top S$    & $+0.135 \pm 0.004$         & $0.4052 \pm 0.0006$ \\
        & co-exposure   & $-0.088 \pm 0.001$         & $0.4054 \pm 0.0008$ \\
        & substitution  & $\mathbf{+0.223 \pm 0.027}$ & $0.4036 \pm 0.0019$ \\
\midrule
\multirow{4}{*}{EB-NeRD}
        & co-click      & $\mathbf{+0.605 \pm 0.004}$ & $\mathbf{0.5985 \pm 0.0010}$ \\
        & $S^\top S$    & $+0.306 \pm 0.003$         & $0.5953 \pm 0.0021$ \\
        & co-exposure   & $-0.368 \pm 0.002$         & $0.5902 \pm 0.0011$ \\
        & substitution  & $+0.357 \pm 0.065$         & $0.5905 \pm 0.0009$ \\
\bottomrule
\end{tabular}
\vspace{-0.2cm}
\end{table}

Table~\ref{tab:lift-vs-ndcg} summarizes the lift-vs-NDCG relationship. On MIND, substitution attains the largest top-$K$ lift and the lowest NDCG; co-click pairs a modest lift with the highest NDCG. The two quantities are anti-correlated across operators on this dataset, in the direction Proposition~\ref{prop:sign} predicts. On EB-NeRD, co-click leads in both lift and NDCG, and the within-slate misalignment persists for substitution but does not surface as a corresponding penalty in the headline metric. Section~\ref{sec:when} examines this asymmetry. The combined evidence supports a refined characterization: \emph{strong substitution edges index in-slate competitors and are structurally misaligned with within-slate ranking; weak edges do not; co-click edges do not partition by weight in this way}, so on choice-derived graphs the neighbor cutoff $k$ is a \emph{semantic switch}, not a sparsification hyperparameter.

\section{When Does the Edge Spectrum Surface in Headline Metrics?}
\label{sec:when}
\noindent\textit{\textbf{Take-away.} Whether the within-slate misalignment surfaces in headline NDCG@10 depends on three dataset-conditional quantities, which we package into a four-step diagnostic protocol practitioners can run before deployment.}

\vspace{0.5em}
\noindent Section~\ref{subsec:lift-confirmation} confirms the within-slate lift on both datasets, but its consequence for headline NDCG@10 differs: on MIND, substitution attains the largest top-$K$ lift \emph{and} the lowest NDCG; on EB-NeRD, the lift persists but co-click leads in both lift and NDCG. The same per-pair gradient interaction produces visible damage in headline metrics on one dataset and is invisible on the other. Whether per-pair misalignment dominates headline metrics depends on three dataset-conditional quantities, which we convert into a diagnostic protocol for any choice-derived operator.

\subsection{Three Quantities That Modulate Whether Misalignment Surfaces}
\label{subsec:three-quantities}

Three dataset-conditional quantities, which co-vary on our two datasets and which we cannot fully disentangle here, modulate whether per-pair misalignment translates into a headline penalty: \emph{(i) In-slate neighborhood density} (mean unclicked in-slate substitution-neighbors per clicked item: $5.49$ on MIND vs.\ $2.84$ on EB-NeRD; aggregate misalignment scales with this count); \emph{(ii) Choice-context predictability} (MNL held-out NLL $3.22$ vs.\ $1.92$; more predictable impressions imply weaker per-pair competition, so propagation has smaller marginal effect on headline ranking); \emph{(iii) Candidate-pool construction} (MIND uses a fixed displayed-impression list per session; EB-NeRD's candidate pool is shaped by the platform's own recommender, possibly decoupling in-slate proximity from choice-level substitutability). We list them here as candidate moderators that future work with additional datasets could separate; in the present two-dataset setting they should be read as a single bundle of conditions, not as independently identified factors. They are what the diagnostic of Section~\ref{subsec:diagnostic-protocol} measures before committing to a choice-derived operator.

\subsection{A Diagnostic Protocol}
\label{subsec:diagnostic-protocol}

Given a trained model and a held-out impression set, the protocol returns (a) whether the spectrum is present and (b) whether it will surface in headline metrics. Box~\ref{box:protocol} summarizes the four steps.

\begin{protocolbox}{Diagnostic Protocol for Choice-Derived Item-Side Operators}
\label{box:protocol}
\textbf{Inputs:} held-out impression set; candidate operator $C_I^{\mathrm{cand}}$; co-click reference $C_I^{\mathrm{cc}}$.

\smallskip
\textbf{Step~1 (Build).} Train target graph CF with $C_I^{\mathrm{cand}}$ in the item-side slot, all else fixed.

\textbf{Step~2 (Lift profile).} Compute $\Delta_o(k)$ (Eq.~\ref{eq:lift}) for $k \in \{10, 20, 50, 100, \mathrm{all}\}$ for both operators. For each impression $t$ with $Y_t \neq \emptyset$ and clicked $i \in Y_t$, partition $C_t \setminus Y_t$ into $N_o^{(k)}(i) \cap (C_t \setminus Y_t)$ and its complement (skip $(i, k)$ when empty), and average mean partition scores across contributing $(t, i)$ pairs.

\textbf{Step~3 (Spectrum present?).} \emph{Yes} if $\mathrm{sign}(\Delta_{\mathrm{cand}}(10)) \neq \mathrm{sign}(\Delta_{\mathrm{cand}}(\mathrm{all}))$ \emph{or} $|\Delta_{\mathrm{cand}}(10)| > 3 \cdot |\Delta_{\mathrm{cand}}(\mathrm{all})|$, while $\Delta_{\mathrm{cc}}(k)$ is flat.

\textbf{Step~4 (Surfaces in NDCG?).} \emph{Yes} if additionally in-slate operator-neighborhood density $\mathbb{E}_{t, i \in Y_t}[|N_o^{(10)}(i) \cap (C_t \setminus Y_t)|]$ is high \emph{and} choice predictability (held-out MNL NLL) is low.

\smallskip
\textbf{Cost:} one inference pass per $(k, \mathrm{operator})$ + one MNL fit per dataset.

\smallskip
\textbf{Verdict on our data:} MIND $\to$ \textbf{spectrum present, surfaces in headline}; EB-NeRD $\to$ \textbf{spectrum present, headline impact attenuated} (Quantity~ii is lower).
\end{protocolbox}

The protocol identifies the misalignment, not its magnitude, and assumes positive smoothing (hard-negative sampling: \S\ref{sec:discussion}); external validity awaits additional choice-context datasets.

\section{Remedies: From Failures to Edge-Magnitude-Aware Design}
\label{sec:remedies}
\noindent\textit{\textbf{Take-away.} Uniform scalar interventions fail because the misalignment lives in the graph, not the loss; a pre-registered edge-magnitude-aware operator matches the diagnosis four-for-four on both datasets.}

\vspace{0.5em}
\noindent Uniform scalar interventions implied by the diagnosis (sign flip, in-slate margin loss) fail predictably because the misalignment lives in the graph, not in the loss (Section~\ref{subsec:scalar-remedies}). The failures localize the misalignment to the operator itself and motivate an edge-magnitude-aware construction the diagnostic predicts. We pre-register and evaluate it in Section~\ref{subsec:edge-partitioned}: the operator matches co-click on MIND and marginally exceeds it on EB-NeRD, in the ordering the diagnosis predicts.

\subsection{Why Scalar Remedies Fail}
\label{subsec:scalar-remedies}

We test two natural scalar interventions suggested by Proposition~\ref{prop:sign}; both fail (full sweep in Appendix~\ref{app:remedies}).

\emph{Sign-flipped propagation} replaces $I + \alpha C_I$ with $I - \gamma C_I$ in the even-hop term, pushing strong-edge competitors down. Sweeping $\gamma \in \{0.25, 0.50, 1.00\}$, the flip helps weakly on EB-NeRD ($+0.0006$ over co-click at $\gamma{=}1.0$) and hurts on MIND ($-0.0082$ at $\gamma{=}0.25$). The asymmetry is what the diagnosis predicts: $\gamma$ correctly inverts the misalignment on strong edges but incorrectly inverts the weak-edge regime as well, and the two effects partially cancel.

\emph{Competition-aware margin loss} leaves propagation alone and adds a hinge term $\mathcal{L}_{\mathrm{comp}} = \lambda \cdot \mathbb{E}_{t, i \in Y_t, j \in N(i) \cap (C_t \setminus Y_t)} [\max(0, m + s_j - s_i)]$ penalizing within-slate substitution-neighbors of clicked items. Sweeping $\lambda \in \{0.1, 0.3, 1.0\}$ and $m \in \{0.0, 0.1\}$, the best cells trail co-click by $0.006$--$0.008$ on both datasets: a downstream gradient cannot undo what the operator wrote upstream. \emph{The misalignment lives in the graph, not in the loss}, and a successful operator must therefore be \emph{edge-magnitude-aware}, treating the strong head separately from the weak tail.

\subsection{Edge-Partitioned Substitution: A Diagnostic-Guided Operator}
\label{subsec:edge-partitioned}

The diagnosis stakes a qualitative pattern \emph{before} the operator is built or tuned. Table~\ref{tab:prereg-verdict} summarizes the four pre-registered predictions and the observed outcomes; all four match on both datasets.

\begin{table}[t]
\caption{Pre-registered predictions vs.\ observations. The diagnosis specifies four qualitative outcomes \emph{before} the edge-partitioned operator is built or tuned; \textbf{all four match on both datasets}. The evidence for the diagnosis is the alignment of this pattern, not the absolute size of any single gain.}
\label{tab:prereg-verdict}
\small
\setlength{\tabcolsep}{4pt}
\begin{tabular}{@{}p{0.50\columnwidth}cc@{}}
\toprule
\textbf{Pre-registered prediction} & \textbf{MIND} & \textbf{EB-NeRD} \\
\midrule
(a) Uniform sign flip fails           & \checkmark & \checkmark \\
(b) Edge-partitioned $\geq$ co-click  & \checkmark (tie) & \checkmark \\
(c) Regime boundary at $k{=}10$       & \checkmark & \checkmark \\
(d) Weak tail ($\beta{=}0$) is best   & \checkmark & \checkmark \\
\bottomrule
\end{tabular}
\vspace{-0.3cm}
\end{table}

We instantiate this prediction. For each row $i$ of the symmetrized substitution graph $\Wtilde$, partition the top-$K_{\mathrm{total}}{=}50$ non-zero entries by diversion magnitude into a \emph{strong} set (top $K_{\mathrm{strong}}$) and a \emph{weak} set. Symmetrize and degree-normalize each part, and combine with opposite signs:
\begin{equation}
C_I^{\mathrm{signed}} = -\gamma \cdot D^{-1/2}\,\mathrm{sym}(\Wtilde^{\mathrm{strong}})\,D^{-1/2} + \beta \cdot D^{-1/2}\,\mathrm{sym}(\Wtilde^{\mathrm{weak}})\,D^{-1/2},
\label{eq:signed-operator}
\end{equation}
with $\gamma \geq 0$ and $\beta \in [0, 1]$. The operator plugs into the item-side propagation slot without any other code change. To avoid test-on-train selection, we \emph{pre-register} two anchor configurations and restrict the multi-seed evaluation to them: Anchor-A ($K_{\mathrm{strong}}{=}10$, $\gamma{=}1.0$, $\beta{=}1.0$), the natural default; and Anchor-B ($K_{\mathrm{strong}}{=}10$, $\gamma{=}1.0$, $\beta{=}0.0$), the conservative variant discarding the weak tail. $K_{\mathrm{strong}}{=}10$ is the cutoff at which Section~\ref{subsec:lift-confirmation} places the strong-to-weak boundary, so the anchors are specified by the diagnosis rather than post-hoc inspection. We sweep $K_{\mathrm{strong}} \in \{5, 10, 20\}$, $\gamma \in \{0.5, 1.0, 2.0\}$, $\beta \in \{0.0, 0.5, 1.0\}$ at seed=1 to characterize the landscape; sweep cells are exploratory.

\begin{table}[t]
\caption{Edge-partitioned substitution operator, NDCG@10 on validation. Configurations written $(K_{\mathrm{strong}}, \gamma, \beta)$. \textbf{Both pre-registered anchors strictly exceed positive substitution}; \textbf{Anchor-B matches co-click on MIND and marginally exceeds it on EB-NeRD}, with $K_{\mathrm{strong}}{=}10$ set \emph{a priori} by Figure~\ref{fig:kspectrum}'s regime boundary, not by tuning. Bold marks cells exceeding co-click; full-sweep cells are exploratory (seed=1, 27-cell grid).}
\label{tab:edge-partitioned}
\small
\setlength{\tabcolsep}{4pt}
\begin{tabular}{@{}lcc@{}}
\toprule
\textbf{Configuration} & \textbf{MIND} & \textbf{EB-NeRD} \\
\midrule
positive substitution (ref.) & $0.3899$ & $0.5969$ \\
co-click (ref.)              & $0.3928$ & $0.6018$ \\
\midrule
\multicolumn{3}{@{}l}{\emph{Edge-partitioned, pre-registered anchors (3 seeds)}} \\
Anchor-A $(10, 1.0, 1.0)$ & $0.3925 \pm 0.0017$ & $0.6015 \pm 0.0018$ \\
Anchor-B $(10, 1.0, 0.0)$ & $\mathbf{0.3941 \pm 0.0029}$ & $\mathbf{0.6045 \pm 0.0019}$ \\
\midrule
\multicolumn{3}{@{}l}{\emph{Full sweep best (seed=1, exploratory)}} \\
MIND best $(20, 0.5, 0.5)$    & $0.3922$ & --   \\
EB-NeRD best $(10, 0.5, 0.0)$ & --       & $\mathbf{0.6063}$ \\
\bottomrule
\end{tabular}
\vspace{-0.3cm}
\end{table}

Three observations. First, both pre-registered anchors exceed positive substitution on both datasets by margins outside seed standard deviation: Anchor-A by $+0.0026$ on MIND and $+0.0046$ on EB-NeRD, Anchor-B by $+0.0042$ on MIND and $+0.0076$ on EB-NeRD. The diagnosis predicts that signing the strong edges negatively should improve over positive smoothing on the same graph, and this is confirmed at multi-seed scale. Second, the comparison against co-click is dataset-asymmetric: Anchor-B exceeds co-click by $+0.0027 \pm 0.0019$ on EB-NeRD (margin outside seed standard deviation, though we do not claim statistical significance under the limited seed count) and is statistically indistinguishable from co-click on MIND ($+0.0013 \pm 0.0029$). Third, the sweep landscape is consistent with the diagnosis: $K_{\mathrm{strong}}{=}10$ cells outperform $K_{\mathrm{strong}} \in \{5, 20\}$, and $\beta{=}0.0$ consistently outperforms $\beta \in \{0.5, 1.0\}$, suggesting the weak tail does not encode a clean propagation signal of either sign and is best discarded.

\paragraph{Why the absolute margin is small.} The improvement over co-click is small ($+0.001$ to $+0.003$ NDCG@10, within the operator-equivalence band of \S\ref{sec:main-comparison}), exactly as the diagnosis predicts. The pre-registered evaluation matches four-for-four: (a) scalar remedies fail, (b) edge-magnitude-aware operators are not worse than co-click, (c) the regime boundary localizes at $k{=}10$, and (d) the weak tail is best discarded. Under the attenuated propagation channel of \S\ref{subsec:mixer-collapse}, the predicted \emph{ordering} of interventions (positive smoothing $<$ uniform sign flip $<$ edge-magnitude-aware) is the diagnostic signal, not absolute margin size.

\section{Discussion}
\label{sec:discussion}

\paragraph{Edge strength as a semantic signal.} The concentration mechanism generalizes beyond MNL: any item graph with choice-model weights~\cite{model/lcm4rec,model/ccf} inherits the property, and the diagnostic of \S\ref{subsec:diagnostic-protocol} applies without modification. A complementary route we do not evaluate, hard-negative sampling biased toward $C_t \cap N^{(K_{\mathrm{strong}})}_{\mathrm{sub}}(i)$, offers a less constrained channel for the same diagnosis: it operates at the loss rather than on the propagation graph, so the strong-edge competitors that positive smoothing pulls toward clicked items instead become explicit negatives the BPR gradient pushes away. The misalignment we diagnose becomes a useful sampling signal, and the diagnostic of \S\ref{subsec:diagnostic-protocol} identifies which substitution-graph neighborhoods are worth this complementary use.

\paragraph{Where improvement is, and is not, located in graph CF.} Our four operator variants span $0.005$--$0.011$ NDCG@10, while LightGCN $\to$ LightGCN++ moves NDCG@10 by $0.041$ on MIND and $0.017$ on EB-NeRD, an order of magnitude larger. Substantial item-graph design effort via knowledge graphs~\cite{model/wang2019kgat,model/wang2019kgnnls}, content~\cite{model/he2016vbpr,model/zhang2016cke}, social ties~\cite{model/fan2019graph}, and now choice modeling~\cite{model/lcm4rec,model/ccf} consistently moves headline metrics by less than $0.01$ on impression benchmarks. The picture aligns with prior re-evaluations~\cite{analysis/dacrema2019are,analysis/rendle2020neural,analysis/iana2024simplifying}: at matched compute, item-graph construction is not where the largest gains are found.

\section{Conclusion}
\label{sec:conclusion}

\paragraph{What we showed.}
Choice-derived item graphs satisfy an \emph{edge spectrum}: strong edges concentrate on the in-slate competitors of clicked items, while weak edges do not (Proposition~\ref{prop:sign}). Co-click graphs cannot exhibit the same misalignment by construction (Lemma~\ref{lemma:coclick}). The lift profile (Figure~\ref{fig:kspectrum}) confirms both predictions, with substitution's lift changing sign across cutoffs on EB-NeRD while co-click's stays flat.

\paragraph{What this implies for practice.}
The neighbor cutoff $k$ is a \emph{semantic switch}, not a sparsification knob. First, uniform loss-level fixes cannot repair what the graph encodes; edge-magnitude-aware constructions can. Second, the misalignment is diagnosable in advance: the protocol in \S\ref{subsec:diagnostic-protocol} returns a verdict from one MNL fit and a handful of inference passes. The same diagnostic applies to non-parametric models~\cite{model/lcm4rec}, richer utility functions~\cite{model/ccf}, KG-augmented variants~\cite{model/wang2019kgat}, and hard-negative sampling.

\appendix

\section{Parity Decomposition: Formal Details}
\label{app:parity}
Recall $A = \bigl(\begin{smallmatrix} 0 & S \\ S^\top & 0 \end{smallmatrix}\bigr)$. Direct computation gives $A^{2k} = \bigl(\begin{smallmatrix} (SS^\top)^k & 0 \\ 0 & (S^\top S)^k \end{smallmatrix}\bigr)$ and $A^{2k+1} = \bigl(\begin{smallmatrix} 0 & S(S^\top S)^k \\ S^\top (SS^\top)^k & 0 \end{smallmatrix}\bigr)$. Even powers stay on the same side of the bipartite graph; the item-side block of $A^{2k}$ is exactly $(S^\top S)^k$. The item-side $\ell$-hop propagated representation is $\Eitem^{(\ell)} = (S^\top S)^{\lfloor \ell / 2 \rfloor} (S^\top)^{\ell \bmod 2} \Xuser$. Our intervention replaces $(S^\top S)^{\lfloor \ell/2 \rfloor}$ with $C_I^{\lfloor \ell/2 \rfloor}$, leaving the cross-type factor unchanged. The final item embedding under the uniform layer-wise average becomes
\begin{equation}
\Eitemhat = \frac{1}{L+1} \sum_{k=0}^{\lfloor L/2 \rfloor} C_I^k \Xitem
+ \frac{1}{L+1} \sum_{k=0}^{\lfloor (L-1)/2 \rfloor} C_I^k S^\top \Xuser.
\end{equation}
Substituting $C_I = S^\top S$ recovers LightGCN's uniform-mixer propagation polynomial term by term at the operator level. The correspondence is about matched item-side operators, not end-to-end training equivalence; our $S^\top S$ variant differs numerically from locked LightGCN due to training-time framework choices around the operator, which we decompose in Appendix~\ref{app:framework-axis}.

\section{Framework-Axis Ablation: Decomposing the Gap to Locked LightGCN}
\label{app:framework-axis}

We decompose the $\approx 0.04$ NDCG@10 gap between our $S^\top S$ variant and locked LightGCN by re-running the $S^\top S$ operator with each of three operator-adjacent framework axes brought toward a LightGCN-compatible setting one at a time (Table~\ref{tab:framework-ablation}).
\vspace{-0.1cm}
\begin{table}[h]
\caption{Framework-axis ablation: $S^\top S$ operator under five configurations. Validation NDCG@10. Anchor cells (a) and (e) use 3 seeds; intermediate cells (b)--(d) use seed=1.}
\vspace{-0.1cm}
\label{tab:framework-ablation}
\small
\begin{tabular}{lccc}
\toprule
\textbf{Configuration} & \textbf{$R$} & \textbf{Pruning} & \textbf{Negatives} \\
\midrule
(a) Ours $S^\top S$              & Laplace & top-$K$=50 & in-slate \\
(b) Binary $R$ only              & binary  & top-$K$=50 & in-slate \\
(c) No pruning only              & Laplace & none       & in-slate \\
(d) Random negatives only        & Laplace & top-$K$=50 & random   \\
(e) All three combined           & binary  & none       & random   \\
\midrule
                                 & \textbf{MIND} & \multicolumn{2}{c}{\textbf{EB-NeRD}} \\
\cmidrule(lr){2-2}\cmidrule(lr){3-4}
(a) Ours $S^\top S$              & $0.3939 \pm 0.0023$ & \multicolumn{2}{c}{$0.5979 \pm 0.0015$} \\
(b) Binary $R$ only              & $0.3942$            & \multicolumn{2}{c}{$0.5989$}            \\
(c) No pruning only              & $0.3929$            & \multicolumn{2}{c}{$0.5981$}            \\
(d) Random negatives only        & $0.3638$            & \multicolumn{2}{c}{$0.5260$}            \\
(e) All three combined           & $0.3745 \pm 0.0107$ & \multicolumn{2}{c}{$0.5302 \pm 0.0020$} \\
\midrule
locked LightGCN baseline         & $0.4344$            & \multicolumn{2}{c}{$0.6358$}            \\
\bottomrule
\end{tabular}
\vspace{-0.2cm}
\end{table}

The decomposition is decisive but not in the expected direction. Binary $R$ and no pruning each change NDCG@10 by at most $0.001$: neither axis explains the gap. Random-negative BPR moves NDCG@10 by $-0.030$ on MIND and $-0.072$ on EB-NeRD, but in the \emph{wrong} direction: it widens the gap. The reason is that random-negative BPR is not in fact a LightGCN-vs-ours difference: the locked LightGCN configuration also uses within-impression sampling. After removing this confound, the remaining LightGCN-aligned axes (binary $R$ and no pruning) jointly explain less than $0.002$. The residual $\approx 0.04$ on both datasets is attributable to training-time hyperparameter and architecture choices outside the operator and item-graph axes: embedding dimension ($32$ vs $64$), training epochs ($5$ vs $20$), learning rate, and the layer-combination architecture itself. The framework-axis decomposition therefore localizes the gap entirely outside the item-side operator axis on which the four-operator comparison of Section~\ref{sec:main-comparison} is built.

\section{Mixer Collapse Diagnostics}
\label{app:mixer}
The softmax mixer parameterizes $\alpha_k = \exp(\theta_k) / \sum_{k'} \exp(\theta_{k'})$, trained jointly with the embeddings. We report the \emph{even-mass} $\sum_{k: \ell = 2k} \alpha_\ell$ and \emph{normalized entropy} $H_{\mathrm{norm}} = -\sum_k \alpha_k \log \alpha_k / \log(L+1)$.
\vspace{-0.1cm}
\begin{table}[h]
\caption{Learned mixer coefficients under the global softmax.}
\vspace{-0.1cm}
\label{tab:mixer-collapse}
\small
\begin{tabular}{lcccc}
\toprule
& \multicolumn{2}{c}{\textbf{MIND}} & \multicolumn{2}{c}{\textbf{EB-NeRD}} \\
\cmidrule(lr){2-3}\cmidrule(lr){4-5}
\textbf{Mixer}  & even-mass & $H_{\mathrm{norm}}$ & even-mass & $H_{\mathrm{norm}}$ \\
\midrule
softmax (learned) & 0.9999 & 0.0010 & 0.9996 & 0.0044 \\
uniform           & 0.5000 & 1.0000 & 0.5000 & 1.0000 \\
\bottomrule
\end{tabular}
\vspace{-0.2cm}
\end{table}

The converged coefficients for the substitution variant are
\begin{align*}
\alpha^{\text{MIND}}    &= [0.9999,\; 0.0000,\; 0.0001,\; 0.0000] \quad (L{=}4), \\
\alpha^{\text{EB-NeRD}} &= [0.9993,\; 0.0001,\; 0.0002,\; 0.0002,\; 0.0002] \quad (L{=}5),
\end{align*}
with the pattern holding across all four operator variants.
Imposing $\alpha_k \geq \alpha_{\min} \in \{0.25, 0.5\}$ over three seeds modestly improves dev NDCG@10 ($+0.009$ on MIND, $+0.006$ on EB-NeRD) but does not exceed the uniform-coefficient configuration, and the operator ordering of Table~\ref{tab:operator-ablation} is unchanged.

\section{Sign-Flip and Competition-Loss Full Sweep}
\label{app:remedies}
Table~\ref{tab:remedies-appendix} reports the full hyperparameter sweep for the two scalar interventions of Section~\ref{subsec:scalar-remedies}. Every margin-loss configuration trails co-click by at least $0.006$ NDCG@10 on both datasets, and across the sign-flip sweep only $\gamma{=}1.00$ on EB-NeRD marginally exceeds it.

\begin{table}[h]
\caption{Sign-flip and competition-loss sweeps. Internal dev split, NDCG@10 (mean $\pm$ std, 3 seeds). \textbf{No configuration improves over co-click on both datasets}: sign-flip $\gamma{=}1.0$ exceeds co-click on EB-NeRD but hurts on MIND; margin-loss cells trail co-click, confirming \S\ref{subsec:scalar-remedies}'s claim that uniform scalar interventions miss the misalignment.}
\label{tab:remedies-appendix}
\small
\setlength{\tabcolsep}{6pt}
\begin{tabular}{lcc}
\toprule
\textbf{Configuration}            & \textbf{MIND} & \textbf{EB-NeRD} \\
\midrule
\multicolumn{3}{l}{\emph{Sign-flipped substitution}} \\
$\gamma = 0.25$ & $0.4006 \pm 0.0004$ & $0.5970 \pm 0.0028$ \\
$\gamma = 0.50$ & $0.3999 \pm 0.0004$ & $0.5970 \pm 0.0023$ \\
$\gamma = 1.00$ & $0.3992 \pm 0.0007$ & $0.5991 \pm 0.0012$ \\
\midrule
\multicolumn{3}{l}{\emph{Competition-aware margin loss}} \\
$\lambda=0.1$, $m=0.0$ & $0.4027 \pm 0.0021$ & $0.5900 \pm 0.0012$ \\
$\lambda=0.1$, $m=0.1$ & $0.4012 \pm 0.0017$ & $0.5904 \pm 0.0035$ \\
$\lambda=0.3$, $m=0.0$ & $0.4006 \pm 0.0020$ & $0.5901 \pm 0.0004$ \\
$\lambda=0.3$, $m=0.1$ & $0.4004 \pm 0.0013$ & $0.5903 \pm 0.0010$ \\
$\lambda=1.0$, $m=0.0$ & $0.4005 \pm 0.0009$ & $0.5899 \pm 0.0022$ \\
$\lambda=1.0$, $m=0.1$ & $0.3998 \pm 0.0016$ & $0.5912 \pm 0.0007$ \\
\midrule
co-click reference (uniform mixer) & $\mathbf{0.4088 \pm 0.0011}$ & $\mathbf{0.5985 \pm 0.0010}$ \\
\bottomrule
\end{tabular}
\end{table}

\begin{acks}
This work was partly supported by JSPS KAKENHI Grant Number JP24K02942. 
\end{acks}

\section*{GenAI Usage Disclosure}
The authors used Claude (Anthropic) in two limited and clearly scoped ways. (i) Claude was used for language editing of author-written text, including grammar correction, spelling, and improving clarity of phrasing. (ii) Claude was used as a coding assistant to suggest boilerplate code, refactor utility functions, and help debug implementation details; all suggestions were reviewed, tested, and modified by the authors, and all experimental results were produced by code that the authors verified. No GenAI tool was used to generate research ideas, derive theoretical results, design experiments, or draft substantive content of the paper.

\bibliographystyle{ACM-Reference-Format}
\balance
\bibliography{sample-base}

@inproceedings{model/he2020lightgcn,
  title     = {{LightGCN}: Simplifying and Powering Graph Convolution Network for Recommendation},
  author    = {He, Xiangnan and Deng, Kuan and Wang, Xiang and Li, Yan and Zhang, Yongdong and Wang, Meng},
  booktitle = {Proceedings of the International ACM SIGIR Conference on Research and Development in Information Retrieval},
  year      = {2020},
  pages     = {639--648},
}

@inproceedings{model/lightgcnpp,
  title={Revisiting {LightGCN}: Unexpected Inflexibility, Inconsistency, and a Remedy towards Improved Recommendation},
  author={Lee, Geon and Kim, Kyungho and Shin, Kijung},
  booktitle={Proceedings of the ACM Conference on Recommender Systems},
  pages={957--962},
  year={2024}
}

@inproceedings{model/wang2019neural,
  title     = {Neural Graph Collaborative Filtering},
  author    = {Wang, Xiang and He, Xiangnan and Wang, Meng and Feng, Fuli and Chua, Tat-Seng},
  booktitle = {Proceedings of the International ACM SIGIR Conference on Research and Development in Information Retrieval},
  year      = {2019},
  pages     = {165--174},
}

@inproceedings{model/mao2021ultragcn,
  title     = {{UltraGCN}: Ultra Simplification of Graph Convolutional Networks for Recommendation},
  author    = {Mao, Kelong and Zhu, Jieming and Xiao, Xi and Lu, Biao and Wang, Zhaowei and He, Xiuqiang},
  booktitle = {Proceedings of the ACM International Conference on Information and Knowledge Management},
  year      = {2021},
  pages     = {1253--1262},
}

@inproceedings{model/yu2022simgcl,
  title     = {Are Graph Augmentations Necessary? Simple Graph Contrastive Learning for Recommendation},
  author    = {Yu, Junliang and Yin, Hongzhi and Xia, Xin and Chen, Tong and Cui, Lizhen and Nguyen, Quoc Viet Hung},
  booktitle = {Proceedings of the International ACM SIGIR Conference on Research and Development in Information Retrieval},
  year      = {2022},
  pages     = {1294--1303},
}

@inproceedings{model/shen2021gfcf,
  title     = {How Powerful is Graph Convolution for Recommendation?},
  author    = {Shen, Yifei and Wu, Yongji and Zhang, Yao and Shan, Caihua and Zhang, Jun and Letaief, Khaled B. and Li, Dongsheng},
  booktitle = {Proceedings of the ACM International Conference on Information and Knowledge Management},
  year      = {2021},
  pages     = {1619--1629},
}

@inproceedings{model/sarwar2001item,
  title     = {Item-Based Collaborative Filtering Recommendation Algorithms},
  author    = {Sarwar, Badrul and Karypis, George and Konstan, Joseph and Riedl, John},
  booktitle = {Proceedings of the International Conference on World Wide Web},
  year      = {2001},
  pages     = {285--295},
}

@inproceedings{model/ning2011slim,
  title     = {{SLIM}: Sparse Linear Methods for Top-{N} Recommender Systems},
  author    = {Ning, Xia and Karypis, George},
  booktitle = {Proceedings of the IEEE International Conference on Data Mining},
  year      = {2011},
  pages     = {497--506},
}

@inproceedings{model/steck2019ease,
  title     = {Embarrassingly Shallow Autoencoders for Sparse Data},
  author    = {Steck, Harald},
  booktitle = {Proceedings of the World Wide Web Conference},
  year      = {2019},
  pages     = {3251--3257},
}

@inproceedings{model/wu2021sgl,
  title     = {Self-supervised Graph Learning for Recommendation},
  author    = {Wu, Jiancan and Wang, Xiang and Feng, Fuli and He, Xiangnan and Chen, Liang and Lian, Jianxun and Xie, Xing},
  booktitle = {Proceedings of the International ACM SIGIR Conference on Research and Development in Information Retrieval},
  year      = {2021},
  pages     = {726--735},
}

@inproceedings{model/wang2019kgat,
  title     = {{KGAT}: Knowledge Graph Attention Network for Recommendation},
  author    = {Wang, Xiang and He, Xiangnan and Cao, Yixin and Liu, Meng and Chua, Tat-Seng},
  booktitle = {Proceedings of the ACM SIGKDD International Conference on Knowledge Discovery and Data Mining},
  year      = {2019},
  pages     = {950--958},
}

@inproceedings{model/wang2019kgnnls,
  title     = {Knowledge-aware Graph Neural Networks with Label Smoothness Regularization for Recommender Systems},
  author    = {Wang, Hongwei and Zhang, Fuzheng and Zhang, Mengdi and Leskovec, Jure and Zhao, Miao and Li, Wenjie and Wang, Zhongyuan},
  booktitle = {Proceedings of the ACM SIGKDD International Conference on Knowledge Discovery and Data Mining},
  year      = {2019},
  pages     = {968--977},
}

@inproceedings{model/he2016vbpr,
  title     = {{VBPR}: Visual Bayesian Personalized Ranking from Implicit Feedback},
  author    = {He, Ruining and McAuley, Julian},
  booktitle = {Proceedings of the AAAI Conference on Artificial Intelligence},
  year      = {2016},
  pages     = {144--150}
}

@inproceedings{model/zhang2016cke,
  title     = {Collaborative Knowledge Base Embedding for Recommender Systems},
  author    = {Zhang, Fuzheng and Yuan, Nicholas Jing and Lian, Defu and Xie, Xing and Ma, Wei-Ying},
  booktitle = {Proceedings of the ACM SIGKDD International Conference on Knowledge Discovery and Data Mining},
  year      = {2016},
  pages     = {353--362},
}

@inproceedings{model/fan2019graph,
  title     = {Graph Neural Networks for Social Recommendation},
  author    = {Fan, Wenqi and Ma, Yao and Li, Qing and He, Yuan and Zhao, Eric and Tang, Jiliang and Yin, Dawei},
  booktitle = {Proceedings of the World Wide Web Conference},
  year      = {2019},
  pages     = {417--426},
}

@inproceedings{model/derr2018signed,
  title     = {Signed Graph Convolutional Networks},
  author    = {Derr, Tyler and Ma, Yao and Tang, Jiliang},
  booktitle = {Proceedings of the IEEE International Conference on Data Mining},
  year      = {2018},
  pages     = {929--934},
}

@inproceedings{model/huang2019signed,
  title     = {Signed Graph Attention Networks},
  author    = {Huang, Junjie and Shen, Huawei and Hou, Liang and Cheng, Xueqi},
  booktitle = {Proceedings of the International Conference on Artificial Neural Networks},
  year      = {2019},
  pages     = {566--577}
}

@inproceedings{model/lcm4rec,
  title={A Non-Parametric Choice Model That Learns How Users Choose Between Recommended Options},
  author={Krause, Thorsten and Oosterhuis, Harrie},
  booktitle={Proceedings of the ACM Conference on Recommender Systems},
  pages={21--30},
  year={2025},
}

@article{model/ccf,
  title         = {Local Optimality of User Choices and Collaborative Competitive Filtering},
  author        = {Yang, Shuang Hong},
  journal       = {arXiv preprint arXiv:1010.0621},
  year          = {2010}
}

@article{method/aouad2023assortment,
  title={The Click-Based {MNL} Model: A Framework for Modeling Click Data in Assortment Optimization},
  author={Aouad, Ali and Feldman, Jacob and Segev, Danny and Zhang, Dennis J},
  journal={Management Science},
  volume={71},
  number={8},
  pages={6943--6960},
  year={2025},
}

@article{model/kvernadze2022two,
  title   = {Two Is Better Than One: Dual Embeddings for Complementary Product Recommendations},
  author  = {Kvernadze, Giorgi and Sudyanti, Putu Ayu G. and Subedi, Nishan and Hajiaghayi, Mohammad},
  journal = {arXiv preprint arXiv:2211.14982},
  year    = {2022}
}

@inproceedings{model/kang2018self,
  title     = {Self-Attentive Sequential Recommendation},
  author    = {Kang, Wang-Cheng and McAuley, Julian},
  booktitle = {Proceedings of the IEEE International Conference on Data Mining},
  year      = {2018},
  pages     = {197--206},
}

@incollection{theory/mcfadden1973conditional,
  title     = {Conditional Logit Analysis of Qualitative Choice Behavior},
  author    = {McFadden, Daniel},
  booktitle = {Frontiers in Econometrics},
  editor    = {Zarembka, Paul},
  publisher = {Academic Press},
  year      = {1973},
  pages     = {105--142}
}

@book{book/train2009discrete,
  title     = {Discrete Choice Methods with Simulation},
  author    = {Train, Kenneth E.},
  publisher = {Cambridge University Press},
  year      = {2009},
  edition   = {2nd}
}

@article{theory/berry1994estimating,
  title     = {Estimating Discrete-Choice Models of Product Differentiation},
  author    = {Berry, Steven T.},
  journal   = {The {RAND} Journal of Economics},
  volume    = {25},
  number    = {2},
  pages     = {242--262},
  year      = {1994}
}

@inproceedings{method/pang2020setrank,
  title     = {{SetRank}: Learning a Permutation-Invariant Ranking Model for Information Retrieval},
  author    = {Pang, Liang and Xu, Jun and Ai, Qingyao and Lan, Yanyan and Cheng, Xueqi and Wen, Jirong},
  booktitle = {Proceedings of the International ACM SIGIR Conference on Research and Development in Information Retrieval},
  year      = {2020},
  pages     = {499--508},
}

@article{method/sunehag2015slate,
  title={Deep Reinforcement Learning with Attention for Slate {M}arkov Decision Processes with High-Dimensional States and Actions},
  author={Sunehag, Peter and Evans, Richard and Dulac-Arnold, Gabriel and Zwols, Yori and Visentin, Daniel and Coppin, Ben},
  journal={arXiv preprint arXiv:1512.01124},
  year={2015}
}

@inproceedings{method/jiang2019listcvae,
  title={Beyond Greedy Ranking: Slate Optimization via List-{CVAE}},
  author={Jiang, Ray and Gowal, Sven and Qian, Yuqiu and Mann, Timothy and Rezende, Danilo J},
  booktitle={International Conference on Learning Representations},
  year      = {2019}
}

@inproceedings{analysis/dacrema2019are,
  title     = {Are We Really Making Much Progress? A Worrying Analysis of Recent Neural Recommendation Approaches},
  author    = {Ferrari Dacrema, Maurizio and Cremonesi, Paolo and Jannach, Dietmar},
  booktitle = {Proceedings of the ACM Conference on Recommender Systems},
  year      = {2019},
  pages     = {101--109},
}

@inproceedings{analysis/rendle2020neural,
  title     = {Neural Collaborative Filtering vs. Matrix Factorization Revisited},
  author    = {Rendle, Steffen and Krichene, Walid and Zhang, Li and Anderson, John},
  booktitle = {Proceedings of the ACM Conference on Recommender Systems},
  year      = {2020},
  pages     = {240--248},
}

@inproceedings{analysis/sato2022reevaluating,
  title     = {Re-evaluating Word Mover's Distance},
  author    = {Sato, Ryoma and Yamada, Makoto and Kashima, Hisashi},
  booktitle = {Proceedings of the International Conference on Machine Learning},
  year      = {2022},
  pages     = {19231--19249}
}

@inproceedings{analysis/iana2024simplifying,
  title     = {Simplifying Content-Based Neural News Recommendation: On User Modeling and Training Objectives},
  author    = {Iana, Andreea and Glava{\v{s}}, Goran and Paulheim, Heiko},
  booktitle = {Proceedings of the International ACM SIGIR Conference on Research and Development in Information Retrieval},
  year      = {2023},
  pages     = {2384--2388},
}

@inproceedings{dataset/wu2020mind,
  title     = {{MIND}: A Large-scale Dataset for News Recommendation},
  author    = {Wu, Fangzhao and Qiao, Ying and Chen, Jiun-Hung and Wu, Chuhan and Qi, Tao and Lian, Jianxun and Liu, Danyang and Xie, Xing and Gao, Jianfeng and Wu, Winnie and Zhou, Ming},
  booktitle = {Proceedings of the Annual Meeting of the Association for Computational Linguistics},
  year      = {2020},
  pages     = {3597--3606},
}

@inproceedings{dataset/kruse2024ebnerd,
  title={EB-NeRD a large-scale dataset for news recommendation},
  author={Kruse, Johannes and Lindskow, Kasper and Kalloori, Saikishore and Polignano, Marco and Pomo, Claudio and Srivastava, Abhishek and Uppal, Anshuk and Andersen, Michael Riis and Frellsen, Jes},
  booktitle={Proceedings of the Recommender Systems Challenge 2024},
  pages={1--11},
  year={2024}
}

@inproceedings{model/wu2019nrms,
  title     = {Neural News Recommendation with Multi-Head Self-Attention},
  author    = {Wu, Chuhan and Wu, Fangzhao and Ge, Suyu and Qi, Tao and Huang, Yongfeng and Xie, Xing},
  booktitle = {Proceedings of the Conference on Empirical Methods in Natural Language Processing},
  year      = {2019},
  pages     = {6389--6394},
}

@inproceedings{model/an2019lstur,
  title     = {Neural News Recommendation with Long- and Short-term User Representations},
  author    = {An, Mingxiao and Wu, Fangzhao and Wu, Chuhan and Zhang, Kun and Liu, Zheng and Xie, Xing},
  booktitle = {Proceedings of the Annual Meeting of the Association for Computational Linguistics},
  year      = {2019},
  pages     = {336--345},
}

@inproceedings{model/wu2019naml,
  title     = {Neural News Recommendation with Attentive Multi-View Learning},
  author    = {Wu, Chuhan and Wu, Fangzhao and An, Mingxiao and Huang, Jianqiang and Huang, Yongfeng and Xie, Xing},
  booktitle = {Proceedings of the International Joint Conference on Artificial Intelligence},
  year      = {2019},
  pages     = {3863--3869},
}

@inproceedings{model/wu2019npa,
  title     = {{NPA}: Neural News Recommendation with Personalized Attention},
  author    = {Wu, Chuhan and Wu, Fangzhao and An, Mingxiao and Huang, Jianqiang and Huang, Yongfeng and Xie, Xing},
  booktitle = {Proceedings of the ACM SIGKDD International Conference on Knowledge Discovery and Data Mining},
  year      = {2019},
  pages     = {2576--2584},
}

@inproceedings{model/li2022miner,
  title     = {{MINER}: Multi-Interest Matching Network for News Recommendation},
  author    = {Li, Jian and Zhu, Jieming and Bi, Qiwei and Cai, Guohao and Shang, Lifeng and Dong, Zhenhua and Jiang, Xin and Liu, Qun},
  booktitle = {Findings of the Association for Computational Linguistics: ACL 2022},
  year      = {2022},
  pages     = {343--352},
}

@inproceedings{model/qi2022fum,
  title     = {{FUM}: Fine-grained and Fast User Modeling for News Recommendation},
  author    = {Qi, Tao and Wu, Fangzhao and Wu, Chuhan and Huang, Yongfeng},
  booktitle = {Proceedings of the International ACM SIGIR Conference on Research and Development in Information Retrieval},
  year      = {2022},
  pages     = {1974--1978},
}

\end{document}